\documentclass[journal]{IEEEtran}

\usepackage[cmex10]{amsmath}
\usepackage{dsfont}
\usepackage{graphicx}
\usepackage{adjustbox}
\usepackage{url, multicol, multirow, color, tikz, enumerate, verbatim, amsfonts,subeqnarray, romannum,adjustbox}
\usepackage{threeparttable}
\usepackage{algorithmic}
\usepackage[ruled,vlined]{algorithm2e}
\usepackage{rotating}
\usepackage{amsmath,bm}
\usepackage{commath}
\usepackage{tikz}
\usepackage{tabularx}
\usepackage{stfloats}
\usepackage{enumitem}   
\usepackage{enumerate}
\usepackage{stmaryrd}
\usepackage{diagbox}
\usepackage{float} 

\usepackage{color, xcolor}
\usepackage{booktabs}
\usepackage {multirow} 
\usepackage{graphics}
\usepackage{amssymb}
\usepackage{eucal}
\usepackage{CJK} 
\usepackage[framemethod=tikz]{mdframed}
\usepackage{amsmath}
\usepackage{subeqnarray}
\usetikzlibrary{patterns}
\usetikzlibrary{shapes.geometric, arrows}
\usetikzlibrary{shapes,arrows,chains}
\usepackage{amsthm}
\usepackage{graphics}
\usepackage{amsfonts, verbatim, amsmath, color, tikz, subeqnarray}
\usetikzlibrary{arrows.meta}
\usepackage{scrextend}
\usepackage{lineno}

\newtheorem{proposition}{Proposition}

\newtheorem{remark}{Remark}

\usepackage{amsmath}

\newcommand{\ren}{\textcolor{black}}
\newcommand{\rteal}{\textcolor{black}}

\begin{document}

\title{An Analysis of Market-to-Market Coordination}


\author{
Weihang Ren, 
\'Alinson S. Xavier,
Fengyu Wang,
Yongpei Guan, and
Feng Qiu

\thanks{Weihang Ren, \'Alinson S. Xavier, and Feng Qiu are with Energy Systems Division, Argonne National Laboratory, Lemont, IL 60439; Fengyu Wang is with New Mexico State University, Las Cruces, NM 88003; Yongpei Guan is with University of Florida, Gainesville, FL 32611.} 
}



\maketitle

\begin{abstract}
The growing usage of renewable energy resources has introduced significant uncertainties in energy generation, \ren{enlarging} challenges for Regional Transmission Operators (RTOs) in managing transmission congestion. To mitigate congestion that affects neighboring regions, RTOs employ a market-to-market (M2M) process \ren{through an iterative method}, in which they exchange real-time security-constrained economic dispatch solutions and communicate requests for congestion relief. While this method provides economic benefits, it struggles with issues like power swings and time delays. To \ren{explore} the \ren{full} potential of M2M enhancements, \ren{in this paper,} we \ren{first} analyze the current M2M \ren{iterative method} practice to better understand its efficacy and identify \ren{places for} improvements. \ren{Then}, we explore enhancements and \ren{develop} an ADMM \ren{method} for the M2M coordination that \ren{optimizes} congestion management. \ren{Specifically, our ADMM method can achieve a minimal cost that is the same as the cost obtained through a centralized model that optimizes multiple markets altogether. Our final case studies,} across a comprehensive set of multi-area benchmark instances, demonstrate the superior performance of the proposed ADMM algorithm for the M2M process. \ren{Meanwhile}, we identify scenarios where the existing M2M process fails to provide solutions \ren{as a by-product}. Finally, \ren{the algorithm is implemented in an open-source package UnitCommitment.jl for easy access by a broader audience.}

\end{abstract}


\section*{Nomenclature}
\noindent\textit{A. Sets}
\begin{labeling}{aa allig}
    \item[$\mathcal{I}$] Set of RTOs
    \item[$\mathcal{B}_i$] Set of buses in RTO $i$
    \item[$\mathcal{G}_i$] Set of generators in RTO $i$
    \item[$\mathcal{G}_i(b)$] Set of generators at bus $b$ in RTO $i$
    \item[$\mathcal{L}$] Set of transmission constraints, $\mathcal{L} = \mathcal{L}_{fg} \cup \mathcal{L}_{nf}$
    \item[$\mathcal{L}_{fg}$] Set of flowgate constraints
    \item[$\mathcal{L}_{nf}$] Set of non-flowgate constraints
\end{labeling}
	\textit{B. Parameters}
\begin{labeling}{aa allig}
    \item[$C_g$] Marginal cost of power provided by generator $g$
    \item[$D_b$] Load \ren{amount} at bus $b$
    \item[$F_\ell$] Power flow capacity of line $\ell$
    \item[$F_{\ell,i}$] Power flow capacity of line $\ell$ for RTO $i$
    \item[$P^{max}_g$] Maximum \ren{generation amount of} generator $g$
    \item[$P^{min}_g$] Minimum \ren{generation amount of} generator $g$
    \rteal{\item[$R_{\ell}$] Relief request sent from RTO $1$ to RTO $2$ on flowgate $\ell, \ \ell \in \mathcal{L}_{fg}$}
    \item[$\delta_{b,\ell}$] Shift factor of bus $b$ on line $\ell$
    \item[$\Delta$] \ren{Electricity amount} interchange between two RTOs
    \item[$\lambda_{\ell,i}$] Dual multiplier for line $\ell$ in constraint \eqref{eqn:ad-consensus-1} ($i=1$) or \eqref{eqn:ad-consensus-2} ($i=2$)
    \item[$\rho$] Penalty factor for ADMM
\end{labeling}
\textit{C. Decision Variables}
\begin{labeling}{aa allig}
    \item[$f_{\ell,i}$] Power flow of line $\ell$ from RTO $i$
    \item[$f_{\ell,i}^j$] Power flow of line $\ell$ from RTO $i$ solved by RTO $j$
    \item[$p_g$] \ren{Generation amount} of generator $g$
    \item[$p_g^j$] \ren{Generation amount} of generator $g$ solved by RTO $j$ 
    \item[$s_{\ell,i}$] Excess variable in flow limit constraint for flowgate $\ell$ in RTO $i$
\end{labeling}

\section{Introduction}


A significant issue in the United States power grid is the unexpected transmission congestion and associated real-time congestion costs caused by ``loop flows.'' The power grid across the country is divided into different regions based on geography, each managed by a Regional Transmission Organization (RTO) or other entity. While each RTO operates its system scheduling and ensuring reliability independently, its power grids are physically interconnected. This interconnection results in loop flows, and unintended power flows through neighboring RTOs, leading to unexpected transmission congestion and real-time congestion costs \cite{liu2022joint}. \ren{For instance, the cost of real-time congestion for the Midcontinent Independent System Operator (MISO) surged to \$3.7 billion in 2022, more than tripling since 2020 \cite{poto2023som}.} 

To mitigate this congestion, market-to-market (M2M) coordination is employed among some neighboring RTOs. However, further efforts are needed to enhance the current iterative process and reduce the congestion costs. In \ren{the remainder of} this introduction, we describe the cause of the M2M congestion and its economic impact, review studies on mitigating this issue, identify the remaining challenges, and outline our contributions to this field.

\subsection{M2M Congestion}

The congestion caused by loop flows is inevitable in the real-time market\rteal{, e}lectricity flows along the \ren{transmission lines following Kirchhoff's Law \cite{conejo2018power}}, regardless of regulatory scheduling. Consequently, power generated within one RTO \ren{can} travels through another RTO's transmission lines before reaching its scheduled destination, as illustrated in Fig.~\ref{fig:loop-flow}. This loop flow consumes part of the neighboring RTO's transmission capacity, potentially causing congestion and forcing the neighboring RTO to re-dispatch the energy through other lines at a higher cost. 
\begin{figure}[ht]
    \centering
    \includegraphics[width=0.6\columnwidth]{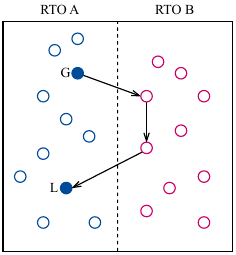}
    \caption{An example of loop flow. The blue circles on the left represent the buses in RTO A, while the pink circles on the right represent the buses in RTO B. The arrows show power generated \ren{from bus G} in RTO A flowing through RTO B before reaching its destination bus \ren{L} in RTO A.}
    \label{fig:loop-flow}
\end{figure}

The congestion costs for RTOs have been significant in recent years, \rteal{emphasizing} the urgency to improve M2M coordination. \ren{Besides the MISO's case described earlier, the t}otal transmission congestion costs for RTOs rose from \$3.7 billion in 2019 to \$12.1 billion in 2022, \ren{and furthermore,} the total congestion costs for the U.S. were estimated to be \ren{up to} \$20.8 billion in 2022 \cite{gs2023report}.

\ren{To mitigate the congestion challenge, s}ince 2003, several neighboring RTO pairs have created Joint Operating Agreements (JOAs), enabling them to efficiently manage constraints affected by both RTOs \cite{spp-miso-joa, pjm-nyiso-joa}. Specifically, the M2M process under JOAs facilitates information exchange between RTOs to reduce overall congestion costs (see details in Section~\ref{sec:m2m}). However, with the increased integration of renewables and fast-\ren{response} units, M2M coordination \ren{can} lead to more oscillation in congestion relief and higher costs.

\subsection{\ren{Related} Literature} \label{sec:review}

Efforts to coordinate multi-region congestion relief in power systems were underway before the development of JOAs. In 1999, a study introduced a protocol for RTOs to communicate information about their net loads outside of their region and to re-dispatch considering loads from other regions \cite{cadwalader1999coordinating}. While this framework managed congestion through energy trading, it did not account for congestion induced by loop flows. To address the problem of loop flows, a later study developed a market-flow-based method, which utilizes market flow to measure the impact of energy market operation on constraints in the coordinated inter-regional congestion management process \cite{luo2014assessment}.

With the creation of JOAs, neighboring RTOs began managing the loop-flow-induced congestion by exchanging shadow prices and relief request information \cite{spp-miso-joa, pjm-nyiso-joa}. This iterative approach helped align border prices and reduce overall congestion costs through the M2M coordination process. However, due to imprecise grid information, the widespread use of fast-\ren{response} units, and delays between dispatch scheduling and resource response, RTOs can over- or under-dispatch, leading to oscillation in power flow, known as ``power swing'' \cite{chen2017address}. To address this issue, researchers proposed an adaptive relief request to stabilize power flow changes from re-dispatch, partially mitigating the oscillation problem. Subsequent research further examined the causes of power swings and raised concerns about market settlement under \ren{the} current M2M coordination \ren{in} \cite{yang2021real}. To improve M2M coordination, \ren{the study in this paper} refined relief request calculations by incorporating predictive flow adjustments, enhancing market flow predictions with data analytics, and improving the adjustable adder logic \ren{in the relief request calculation}.

As the complexity of modern electricity markets grows, various decentralized methods for multi-region coordination have been investigated, addressing challenges in economic dispatch, optimal power flow, and unit commitment problems \cite{wang2017distributed, guo2017coordinated, ji2017multi}. Among these methods, the alternative-direction method of multipliers (ADMM) has proven effective in power system applications. It has demonstrated reliable performance in solving large instances efficiently, without significant convergence or numerical issues \cite{xavier2024decomposable}. However, while the M2M coordination process is suitable for decentralized algorithms, \ren{as compared to other applications,} it is crucial to uphold data privacy between markets \cite{khanabadi2017market}.

\subsection{Challenge and Contribution}
\ren{Recent literature has raised concerns about the power swing issue inherent in the current M2M coordination process, emphasizing the urgent need for solutions. While some efforts have concentrated on refining relief request calculations to facilitate implementation within existing systems, a significant gap persists. Specifically, there is a lack of open-source tools that thoroughly evaluate the M2M process for congestion relief, and considerable opportunities remain to further optimize this process in light of the substantial congestion costs faced by many RTOs.}

To address this gap, this paper provides an analysis of the current M2M process using benchmark cases and proposes a decentralized algorithm to more effectively mitigate congestion. This approach aims to improve the evaluation and management of congestion relief, considering both the limitations of existing methods and the potential benefits of decentralized algorithms. The main contributions are summarized as follows:
\begin{itemize}
    \item[1.] Formalize the M2M coordination problem and develop a centralized formulation to determine a lower bound for the problem;
    
    \item[2.] Assess the current M2M coordination process using realistic benchmark instances, highlighting typical issues and inefficiencies;

    \item[3.] Propose a decentralized method for M2M coordination using ADMM, guaranteeing convergence to the lower bound provided by the centralized method, \ren{so as to achieve an optimal solution};

    \item[4.] Implement the existing M2M coordination algorithm and propose centralized and decentralized formulations as open-source tools.
\end{itemize}


The remaining part of the paper is organized as follows. Section~\ref{sec:central} presents the centralized economic dispatch model for two RTOs that provides a lower bound for congestion mitigation. Section~\ref{sec:m2m} describes the iterative method used in the current M2M coordination process in practice. Section~\ref{sec:admm} proposes an ADMM formulation that solves the M2M coordination in a distributed fashion. Finally, Section~\ref{sec:experiments} demonstrates our implementation of proposed formulations on customized benchmark instances and evaluates the performance of the current iterative method and proposed ADMM solution.

\section{Centralized Formulation} \label{sec:central}

We first formalize the M2M coordination problem and present its model from a centralized perspective \ren{for the real-time market clearance}, where the security-constrained economic dispatch (SCED) problems of two neighboring RTOs are solved together. By sharing all network information, this approach aims to achieve a minimum total cost.

The M2M coordination problem for congestion is a multi-area SCED in the real-time market. Each area has its own set of loads and generators\ren{, and t}he loads in each area are served by its generators. \ren{Besides this, t}he two areas share transmission lines due to their interconnection. The primary decisions in the \ren{centralized} multi-area SCED are the power output of each generator and the allocation of transmission capacity for each region.

The centralized model is formulated by integrating the networks of both RTOs and \ren{ensuring the load balance for the combined area in the integrated SCED} \rteal{model}:
\begin{subeqnarray} \label{model:cen}
    &\hspace{-0.5in}\min & 
    \sum_{g \in \mathcal{G}_1} C_g p_g + \sum_{g \in \mathcal{G}_2} C_g p_g \slabel{eqn:cen-obj}\\
    &\hspace{-0.5in}\mbox{s.t.}
    & P_g^{min} \le p_g \le P_g^{max}, \ \forall g \in \mathcal{G}_1 \cup \mathcal{G}_2, \slabel{eqn:cen-gen} \\
    && \sum_{g \in \mathcal{G}_1} p_g = \sum_{b \in \mathcal{B}_1} D_b + \Delta, \slabel{eqn:cen-balance-1} \\
    && \sum_{g \in \mathcal{G}_2} p_g = \sum_{b \in \mathcal{B}_2} D_b - \Delta, \slabel{eqn:cen-balance-2} \\
    && f_{\ell,1} = \sum_{b \in \mathcal{B}_1} \delta_{b,\ell} \left( \sum_{g \in \mathcal{G}_{1}(b)} p_g - D_b \right), \ \forall \ell \in \mathcal{L}, \slabel{eqn:cen-flow-1} \\
    && f_{\ell,2} = \sum_{b \in \mathcal{B}_2} \delta_{b,\ell} \left( \sum_{g \in \mathcal{G}_{2}(b)} p_g - D_b \right), \ \forall \ell \in \mathcal{L}, \slabel{eqn:cen-flow-2} \\
    && -F_\ell \le f_{\ell,1} + f_{\ell,2} \le F_\ell, \ \forall \ell \in \mathcal{L}, \slabel{eqn:cen-flow-limit}
\end{subeqnarray}
where the objective \eqref{eqn:cen-obj} minimizes the total generation cost across the integrated network, constraint \eqref{eqn:cen-gen} ensures each generator operates within its output limits, constraints \eqref{eqn:cen-balance-1} and \eqref{eqn:cen-balance-2} enforce power balance within each RTO, considering a predetermined power interchange $\Delta$, constraints \eqref{eqn:cen-flow-1} and \eqref{eqn:cen-flow-2} calculate the power flow from each RTO on each transmission line, and constraint \eqref{eqn:cen-flow-limit} restricts that the total flow on each line does not exceed its capacity. 

This centralized formulation offers a clear advantage over treating the two RTOs separately or combining their networks into one. Unlike handling individual SCED problems separately for each region, the centralized method integrates both networks, accounting for the total flow on transmission lines and optimizing congestion management and cost minimization across the interconnected RTOs. Besides, while an ideal solution for M2M coordination might involve combining two RTOs in one SCED, allowing any generator to serve any load across two RTOs. However, this combined SCED requires sharing network topology, generator profiles, and load data. This sharing compromises privacy and is incompatible with the separate operation\ren{s} of RTOs, providing limited insights into the M2M process. In contrast, the centralized method \ren{described in this section} considers the needs within each RTO and bridges them together, providing \ren{the best possible} optimal generation dispatch and flow allocation \ren{for the current market framework. Accordingly, w}e have the following proposition holds.
\begin{proposition}\label{prop:zzz}
    Let $z_0$ be the optimal cost of the combined SCED for neighboring RTOs, $z_1$ \ren{be} the optimal cost of the centralized method \ren{described in formulation \eqref{model:cen}}, and $z_2$ \ren{be} the cost from any M2M coordination. Then, \ren{we have}
    \begin{align}
        z_0 \le z_1 \le z_2.
    \end{align}
\end{proposition}

\begin{proof}
   Any solution from the centralized method is also feasible for the combined SCED, so $z_0 \le z_1$. Furthermore, since M2M solutions are feasible for the centralized model \ren{\eqref{model:cen}. Thus}, $z_1 \le z_2$ \ren{and the conclusion holds}.
\end{proof}

\begin{remark}
    \ren{Based on Proposition~\ref{prop:zzz}, t}he centralized model provides a lower bound for the M2M coordination problem.
\end{remark}

\section{Current Iterative Method} \label{sec:m2m}

Currently, several RTOs, \ren{including MISO, PJM, SPP, and NYISO,} employ an iterative method specified in JOAs for M2M coordination \ren{\cite{spp-miso-joa, pjm-nyiso-joa}}. This approach involves identifying flowgates--binding constraints that could become congestion points. A transmission constraint is considered an M2M flowgate if it is operated by one RTO but significantly affected by another, in a base or a contingency-specific scenario. The RTO responsible for monitoring and control of a flowgate is known as ``Monitoring RTO'' (MRTO), while the RTO influencing the flowgate is termed ``Non-monitoring RTO'' (NMRTO). The MRTO and NMRTO exchange shadow prices, relief requests, and other relevant information iteratively to manage and alleviate congestion.



\rteal{As described in} \cite{spp-miso-joa} and \cite{chen2017address}, \rteal{for the initial step, MRTO and NMRTO are given a portion of the physical limit of each flowgate. That is, for each $\ell \in \mathcal{L}_{fg}$, we have $F_{\ell,1}$ allocated to MRTO and $F_{\ell,2}$ allocated to NMRTO, with $F_{\ell,1} + F_{\ell,2} = F_{\ell}$, by assuming RTO $1$ is MRTO and RTO $2$ is NMRTO. Both RTOs initially solve their respective RTO optimization problems with these given limits for the flowgates and obtain the corresponding shadow prices $\lambda_{\ell,1}$ and $\lambda_{\ell,2}$ for each flowgate $\ell \in \mathcal{L}_{fg}$. Starting from here, MRTO solves the following problem:}
\begin{subeqnarray} \label{model:iter}
    &\hspace{-0.5in}\min & 
    \sum_{g \in \mathcal{G}_1} C_g p_g + \lambda_{\ell,2} s_{\ell,1} \slabel{eqn:iter-obj}\\
    &\hspace{-0.5in}\mbox{s.t.}
    & P_g^{min} \le p_g \le P_g^{max}, \ \forall g \in \mathcal{G}_1, \slabel{eqn:iter-gen} \\
    && \sum_{g \in \mathcal{G}_1} p_g = \sum_{b \in \mathcal{B}_1} D_b, \slabel{eqn:iter-balance} \\
    && \left| \sum_{b \in \mathcal{B}_1} \delta_{b,\ell} \left( \sum_{g \in \mathcal{G}_{1}(b)} p_g - D_b \right) \right| \le F_\ell, \nonumber\\ 
    &&\hspace{1.9in} \forall \ell \in \mathcal{L}_{nf}, \slabel{eqn:iter-flow-nf} \\
    && \left| \sum_{b \in \mathcal{B}_1} \delta_{b,\ell} \left( \sum_{g \in \mathcal{G}_{1}(b)} p_g - D_b \right) \right| \le F_{\ell\rteal{,1}} + s_{\ell,1}, \nonumber\\ 
    &&\hspace{1.9in} \forall \ell \in \mathcal{L}_{fg}. \slabel{eqn:iter-flow-fg}
\end{subeqnarray}
\ren{By solving the above model \eqref{model:iter}, the shadow price for the flowgate in RTO $1$, i.e., $\lambda_{\ell,1}$, is updated. Meanwhile, \rteal{because MRTO monitors and controls the flowgate, it calculates relief requests, asking for more physical limit from NMRTO. For instance}, the relief request is calculated as follows \cite{chen2017address}:
\begin{align} \label{eqn:relief}
    \rteal{R_{\ell} = (f_{\ell,1} + f_{\ell,2}) - F_{\ell}} + Adder, \quad \forall \ell \in \mathcal{L}_{fg},
\end{align}
where
$Adder$ is an additional amount depending on the scenario. When MRTO is binding and NMRTO shadow price is cheaper, an $Adder$, up to 20\% of the flow limit, is added to accelerate the convergence of shadow prices. After this, MRTO sends its shadow prices for the flowgates and relief requests to NMRTO.} 

The NMRTO then solves its SCED considering the MRTO's shadow prices and relief requests. \rteal{Specifically, the NMRTO solves the following problem:}
\begin{subeqnarray} \label{model:iter-n}
    &\hspace{-0.1in}\min & 
    \sum_{g \in \mathcal{G}_2} C_g p_g + \lambda_{\ell,1} s_{\ell,2} \slabel{eqn:iter-obj-n}\\
    &\hspace{-0.1in}\mbox{s.t.}
    & P_g^{min} \le p_g \le P_g^{max}, \ \forall g \in \mathcal{G}_2, \slabel{eqn:iter-gen-n} \\
    && \sum_{g \in \mathcal{G}_2} p_g = \sum_{b \in \mathcal{B}_2} D_b, \slabel{eqn:iter-balance-n} \\
    && \left| \sum_{b \in \mathcal{B}_2} \delta_{b,\ell} \left( \sum_{g \in \mathcal{G}_{2}(b)} p_g - D_b \right) \right| \le F_\ell, \ \forall \ell \in \mathcal{L}_{nf}, \slabel{eqn:iter-flow-nf-n} \\
    && \left| \sum_{b \in \mathcal{B}_2} \delta_{b,\ell} \left( \sum_{g \in \mathcal{G}_{2}(b)} p_g - D_b \right) \right| \le (F_{\ell,2} - R_{\ell}) + s_{\ell,2}, \nonumber\\ 
    &&\hspace{2.2in} \forall \ell \in \mathcal{L}_{fg}. \slabel{eqn:iter-flow-fg-n}
\end{subeqnarray}
\rteal{By solving the above model \eqref{model:iter-n}, the shadow price for the flowgate in RTO $2$, i.e., $\lambda_{\ell,2}$, is updated. If $\lambda_{\ell,2} < \lambda_{\ell,1}$, this indicates that} NMRTO can mitigate the congestion at a lower cost than MRTO. \rteal{In this case,} NMRTO adjusts its flow on the flowgate, and sends its shadow prices back to MRTO. \rteal{If $\lambda_{\ell,2} \ge \lambda_{\ell,1}$, NMRTO still sends its updated shadow prices to MRTO, allowing MRTO to adjust the relief request.} This process repeats until the shadow prices converge. 


While this iterative method aims to use the relatively cheaper units between the two RTOs to alleviate congestion until marginal costs for relief converge, issues like ``power swing'' can arise due to improper relief request amounts and scheduling delays, as reviewed in Section~\ref{sec:review}. We will simulate this iterative method on realistic instances and report the performance in Section~\ref{sec:experiments}.

\section{An ADMM Approach for M2M} \label{sec:admm}

To leverage the advantage of distributed optimization on M2M coordination, we develop an ADMM algorithm for congestion relief between two RTOs. 

ADMM decomposes a complex optimization problem into smaller subproblems and solves them distributively using Lagrangian multipliers, which enhances efficiency. This algorithm guarantees convergence for convex optimization problems like SCED, is easy to implement, highly flexible, and has demonstrated success in various applications \cite{xavier2024decomposable}. Crucially, ADMM retains most information within subproblems, addressing the privacy concerns of RTOs.

For M2M coordination, the centralized formulation \eqref{model:cen} can be split into two subproblems, each for one RTO. The ADMM algorithm is then applied to solve these subproblems. The ADMM \ren{master} formulation is as follows:

\begin{subeqnarray} \label{model:admm}
    &\hspace{-0.5in}\min & 
    \sum_{g \in \mathcal{G}_1} C_g p_g^1 + \sum_{g \in \mathcal{G}_2} C_g p_g^2 \slabel{eqn:ad-obj}\\
    &\hspace{-0.5in}\mbox{s.t.}
    & P_g^{min} \le p_g^1 \le P_g^{max}, \ \forall g \in \mathcal{G}_1, \slabel{eqn:ad-gen-1} \\
    && P_g^{min} \le p_g^2 \le P_g^{max}, \ \forall g \in \mathcal{G}_2, \slabel{eqn:ad-gen-2} \\
    && \sum_{g \in \mathcal{G}_1} p_g^1 = \sum_{b \in \mathcal{B}_1} D_b + \Delta, \slabel{eqn:ad-balance-1} \\
    && \sum_{g \in \mathcal{G}_2} p_g^2 = \sum_{b \in \mathcal{B}_2} D_b - \Delta, \slabel{eqn:ad-balance-2} \\
    && f_{\ell,1}^1 = \sum_{b \in \mathcal{B}_1} \delta_{b,\ell} \left( \sum_{g \in \mathcal{G}_1(b)} p_g^1 - D_b \right), \ \forall \ell \in \mathcal{L}, \slabel{eqn:ad-flow-1} \\
    && f_{\ell,2}^2 = \sum_{b \in \mathcal{B}_2} \delta_{b,\ell} \left( \sum_{g \in \mathcal{G}_2(b)} p_g^2 - D_b \right), \ \forall \ell \in \mathcal{L}, \slabel{eqn:ad-flow-2} \\
    && -F_\ell \le f_{\ell,1}^1 + f_{\ell,2}^1 \le F_\ell, \ \forall \ell \in \mathcal{L}_{fg}, \slabel{eqn:ad-fg-flow-limit-1} \\
    && -F_\ell \le f_{\ell,1}^2 + f_{\ell,2}^2 \le F_\ell, \ \forall \ell \in \mathcal{L}_{fg}, \slabel{eqn:ad-fg-flow-limit-2} \\
    && -F_{\ell,1} \le f_{\ell,1}^1 \le F_{\ell,1}, \ \forall \ell \in \mathcal{L}_{nf}, \slabel{eqn:ad-nf-flow-limit-1} \\
    && -F_{\ell,2} \le f_{\ell,2}^2 \le F_{\ell,2}, \ \forall \ell \in \mathcal{L}_{nf}, \slabel{eqn:ad-nf-flow-limit-2} \\
    && f_{\ell,1}^1 = f_{\ell,1}^2, \ \forall \ell \in \mathcal{L}_{fg}, \slabel{eqn:ad-consensus-1} \\
    && f_{\ell,2}^1 = f_{\ell,2}^2, \ \forall \ell \in \mathcal{L}_{fg}, \slabel{eqn:ad-consensus-2}
\end{subeqnarray}
where the objective \eqref{eqn:ad-obj} and constraints \eqref{eqn:ad-gen-1}-\eqref{eqn:ad-nf-flow-limit-2} are analogous to those in model~\eqref{model:cen}, with decision variables separated into two RTOs marked by their superscripts. \ren{Note that the transmission lines are separated into flowgate set and non-flowgate set, i.e., $\mathcal{L} = \mathcal{L}_{fg} \cup \mathcal{L}_{nf}$.} Constraints \eqref{eqn:ad-consensus-1} and \eqref{eqn:ad-consensus-2} are consensus equations ensuring the same variable\ren{s} from two RTOs are equal.

\ren{Now}, we dualize and penalize the consensus constraints \eqref{eqn:ad-consensus-1} and \eqref{eqn:ad-consensus-2} to create the augmented Lagrangian relaxation \ren{(subproblems)}:
\begin{subeqnarray} \label{model:alr}
    &\hspace{-0.5in} \min\limits_{p^1,p^2,f^1,f^2} & \sum_{g \in \mathcal{G}_1} C_g p_g^1 + \sum_{g \in \mathcal{G}_2} C_g p_g^2 \nonumber \\
    && + \sum_{\ell \in \mathcal{L}} \left[ \lambda_{\ell,1} (f_{\ell,1}^1 - f_{\ell,1}^2) + \lambda_{\ell,2} (f_{\ell,2}^1 - f_{\ell,2}^2) \right] \nonumber \\
    && + \frac{\rho}{2} \sum_{\ell \in \mathcal{L}} \left[ (f_{\ell,1}^1 - f_{\ell,1}^2)^2 + (f_{\ell,2}^1 - f_{\ell,2}^2)^2 \right] \slabel{eqn:alr-obj}\\
    &\hspace{-0.5in}\mbox{s.t.}
    & \eqref{eqn:ad-gen-1}, \eqref{eqn:ad-balance-1}, \eqref{eqn:ad-flow-1}, \eqref{eqn:ad-fg-flow-limit-1}, \eqref{eqn:ad-nf-flow-limit-1}, \slabel{eqn:alr-cons-1} \\
    && \eqref{eqn:ad-gen-2}, \eqref{eqn:ad-balance-2}, \eqref{eqn:ad-flow-2}, \eqref{eqn:ad-fg-flow-limit-2}, \eqref{eqn:ad-nf-flow-limit-2}, \slabel{eqn:alr-cons-2}
\end{subeqnarray}
where $\lambda = (\lambda_{\ell,1}, \lambda_{\ell,2})$ represents dual multipliers for consensus constraints. As a shorthand, let $L_{\rho}(f^1, f^2, \lambda)$ denote the objective function \eqref{eqn:alr-obj}, where $f^i=(f_{\ell,1}^i, f_{\ell,2}^i)$ represents the power flow variables for RTO $i$ for $i=1,2$. Constraints \eqref{eqn:alr-cons-1} and \eqref{eqn:alr-cons-2} only contain constraints for RTO 1 and 2, respectively. 

\ren{In the above formulation \eqref{model:alr}, since} $p^1$ and $p^2$ are only used by their respective RTO\ren{s}, \ren{these two variables} are omitted \ren{in the representation of $L_\rho (\cdot)$} for simplicity when addressing this relaxation \ren{in the rest of the paper}. Note that when $f^1$ \ren{or} $f^2$ is fixed at some $\bar{f}$, $L_{\rho}(\bar{f}, f^2, \lambda)$ \ren{or} $L_{\rho}(f^1, \bar{f}, \lambda)$ is a problem with constraints only for RTO 2 \ren{or} RTO 1. 

This relaxation \ren{\eqref{model:alr}} can be solved using the ADMM algorithm described in Algorithm~\ref{algo:admm}, with the iteration number $k$ indicated by the superscript in parentheses. The stopping criteria first check the solution feasibility by comparing the global residual with the minimum residual requirement, where the global residual is calculated as the sum of the difference between the solutions of $f^1/f^2$ and their target values $\bar{f}^1/\bar{f}^2$, \ren{i.e., $|f^1-\bar{f^1}| + |f^2-\bar{f^2}|$}. \ren{Once the global residual is smaller than the minimum residual threshold, indicating that the solution has reached the feasibility requirement}, the stopping criteria further evaluate whether the absolute change in the objective function is below a specified minimum improvement threshold. The algorithm terminates when both criteria are satisfied.

\begin{algorithm}
\caption{ADMM for M2M Coordination}
\label{algo:admm}
\begin{algorithmic}[1]
\STATE \textbf{Initialize} starting point $\bar{f}^{(0)}$; let $\lambda^{(0)} \leftarrow 0, k \leftarrow 0$;
\WHILE{stopping criteria are not satisfied}
    \STATE RTO 1 solves $\min \ L_{\rho}(f^{1}, \bar{f}^{(k)}, \lambda^{(k)})$ subject to \eqref{eqn:alr-cons-1} and gets $f^{1,(k+1)}$;
    \STATE RTO 2 solves $\min \ L_{\rho}(\bar{f}^{(k)}, f^{2}, \lambda^{(k)})$ subject to \eqref{eqn:alr-cons-2} and gets $f^{2,(k+1)}$;
    \STATE $\bar{f}^{(k+1)} \leftarrow (f^{1,(k+1)} + f^{2,(k+1)}) / 2$;
    \STATE $\lambda^{(k+1)} \leftarrow \lambda^{(k)} + \ren{\rho (f^{1,(k+1)} - \bar{f}^{(k+1)}) + \rho (f^{2,(k+1)} - \bar{f}^{(k+1)})}$;
    \STATE $k \leftarrow k + 1$;
\ENDWHILE
\end{algorithmic}
\end{algorithm}




\begin{proposition} \label{prop:convergence}
    The ADMM \ren{framework} \eqref{model:admm} \ren{and \eqref{model:alr} implemented in Algorithm \ref{algo:admm}} converge\ren{s} to an optimal solution for both markets \ren{as shown in the centralized model \eqref{model:cen},} and meanwhile the \ren{corresponding} LMP price\ren{s} converge to \ren{the} price\ren{s} that \ren{are} optimal for each market \ren{if the optimal solution and the corresponding dual value are unique}, albeit the price\ren{s} for both markets might not be the same at the shared buses.
\end{proposition}

\begin{proof}
    The ADMM formulation for M2M coordination in \eqref{model:admm} is a linear program with continuous variables, making it a convex optimization problem. It is \ren{shown in \cite{boyd2011distributed}} that the ADMM algorithm guarantees convergence for convex problems. Furthermore, since our ADMM formulation  \eqref{model:admm} is derived directly from the centralized model \eqref{model:cen}, it inherits the properties of the centralized model, ensuring convergence to an optimal solution for both markets. \ren{The corresponding LMP prices are associated dual values of the optimal solution for each market, which converge to the optimal dual value if the optimal value is unique.} These prices are optimal for each market, although they may not coincide with the shared buses.
\end{proof}


\section{Computational Experiments} \label{sec:experiments}

This section begins by detailing the construction process for M2M coordination instances with various congestion scenarios. We then provide a comprehensive description of the implementation of both the current iterative method and our proposed ADMM algorithm. For each method, we evaluate their performance in mitigating congestion and discuss the observed issues and practical implications.

All algorithms \ren{were} implemented in Julia 1.10. The SCED framework \ren{was} constructed using UnitCommitment.jl \cite{xavier2024unitcommitment}, modified with JuMP 1.22, and solved using Gurobi 11.0. For ADMM, each instance is solved sequentially, with one process allocated per RTO. Inter-process communication is managed via MPI. The experiments \ren{were} conducted on a computer powered by an Intel Xeon W with 192GB of RAM and an AMD Radeon Pro W5500X.

\subsection{Instance Construction}
Currently, there is a lack of benchmark instances tailored for M2M coordination studies. To address this gap, we developed an automated process to construct an M2M instance from any power system instance that follows the ``common data format'' defined in UnitCommitment.jl, including flowgate identification. 

This process involves three steps: decomposing the given power system network, identifying the most congested flowgate, and splitting the flow capacities into two RTO markets. 

\subsubsection{Decomposing the given power system network}
\ren{T}o simulate the topology of two neighboring RTOs, we create a directed graph based on the given buses and transmission lines. This graph is \ren{first} partitioned into two interconnected smaller networks using Metis, \ren{an advanced graph partition library \cite{karypis1998software}}. Buses in one partition are designated as RTO 1, and the remaining buses are assigned to RTO 2. We then make two copies of this labeled power grid. For the first copy, we remove the load and generator information related to RTO 2 buses and halve the total grid reserve requirements. Similarly, we adjust the second copy by removing the load and generator information \ren{related to} RTO 1 and halving the reserve requirements.

\subsubsection{Identifying the most congested flowgate}
\ren{W}e determine the flowgate based on shift factors and market flows. Transmission lines originating from one RTO but influenced by a generator with a shift factor greater than 5\% from the other RTO are marked as potential flowgates. To obtain market flows, we solve the centralized model \eqref{model:cen} on this network, which provides the flows from two RTOs on each line. We rank potential flowgates by their congestion ratio. The congestion ratio for line $\ell$ is calculated as follows:
\begin{equation}
    \mbox{\it Congestion Ratio} = f_{\ell, 1} f_{\ell, 2} / |f_{\ell, 1} + f_{\ell, 2}|.
\end{equation}
A higher congestion ratio indicates greater congestion and a more balanced contribution of flows from both RTOs. The transmission line with the highest congestion ratio is selected as the flowgate of interest. 

\subsubsection{Allocate flow capacities to the two RTOs}
\ren{To} allocate flow capacities to the two RTOs for each line, \ren{f}or transmission lines that are not the flowgate, if there is no explicit limit, we assign a sufficiently large limit to both RTOs. If both market flows are below half of this limit, we allocate half of the limit to each RTO. If one RTO uses more than half of the limit, we assign 110\% of its market flow to that RTO, and any remaining capacity goes to the other RTO if it exceeds its market flow. For the flowgate, each RTO receives half of the original flow limit. 

Th\ren{e above} approach generates a ``standard instance'' file suitable for M2M coordination studies from \ren{a} given power system instance.

Additionally, we create two more instances \ren{for each given power system network} to explore potential issues \ren{faced by different M2M coordination methods}. One is a ``lower-limit instance'', which reduces the flowgate's total capacity by 5\%, thereby challenging the M2M coordination process with a tighter flowgate capacity. The other is an ``opposite-flow instance'', which selects a flowgate with a low negative congestion ratio, indicating opposite flow directions from the two RTOs on this line, and uses the total market flow as the line capacity.

\ren{Finally, i}n our case study, we generated M2M instances using six popular power system instances \ren{from MATPOWER test cases compiled in \cite{xavier2024unitcommitment}}, ranging from 1,951 to 6,468 buses. The network information is summarized in Table~\ref{tab:network}.

\begin{table}[ht]
\centering
\caption{Instance Summary}
\begin{tabular}{cccc}\toprule
Network     & Buses & Generators & Lines \\\midrule
case1951rte & 1,951 & 390        & 2,596 \\
case2383wp  & 2,383 & 323        & 2,896 \\
case2868rte & 2,868 & 596        & 3,808 \\
case3120sp  & 3,120 & 483        & 3,693 \\
case3375wp  & 3,374 & 590        & 4,161 \\
case6468rte & 6,468 & 1,262      & 9,000 \\\bottomrule
\end{tabular}
\label{tab:network}
\end{table}

\subsection{Evaluation of Iterative Method} \label{sec:iter}

The iterative method \ren{is} currently applied in practice \ren{ for MISO and other RTOs. Al}though successful in saving millions in congestion costs over the years, \ren{it} faces challenges in effectiveness and reliability. In this section, we describe our implementation of the iterative method, report the numerical results from the simulation, and highlight common issues encountered.

We implement the iterative method for M2M coordination following the process described in Section~\ref{sec:m2m}, using the \ren{generated} M2M instances \ren{for testing}. The SCED for each RTO during each iteration is solved through a modified UnitCommitment.jl model. Specifically, \ren{as described in model \eqref{model:iter},} the transmission limit for the flowgate $\ell$ in RTO $i$ is modified as
\begin{align}
    \left|\sum_{b \in \mathcal{B}_i} \delta_{b,\ell} \left ( \sum_{g \in \mathcal{G}_i} p_g^i - D_b \right ) \right| \le F_{\ell \rteal{, i}} + s_{\ell,i},
\end{align}
where the excess variable $s_{\ell,i}$ is added to allow uncontrollable flows. This excess variable is then penalized in the objective function using the other RTO's shadow price, ensuring that this RTO's marginal mitigation cost does not exceed the other one's. Additionally, we observed that when the flow direction on the flowgate is negative, the relief request calculated by the formula \eqref{eqn:relief} might be negative. Therefore, we modified the relief request calculation to \ren{be}
\begin{align} \label{eqn:relief2}
    \rteal{R_{\ell} = \abs{\abs{f_{\ell,1} + f_{\ell,2}} - F_{\ell}}} + Adder,
\end{align}
where all values are determined in the same way as in \eqref{eqn:relief}. This modified formula guarantees the relief request sent by MRTO is positive. Furthermore, the maximum iteration number is set to $10$, so the algorithm will terminate either when the shadow prices from the two RTOs converge or when the maximum iteration number is reached.

The simulation results of the iterative method on the M2M instances are reported in Table~\ref{tab:iterative}. \ren{In the table, the ``Network'' column provides the instance ID.} The instances are categorized into two groups: the upper six are standard instances used to validate the current iterative method, while the lower seven, marked by the suffix ``-of'' and ``-ll'' in their IDs for ``opposite-flow'' and ``lower-limit'' instances respectively, are selected to illustrate the challenges encountered with the current approach. The ``Interchange Ratio'' column represents the portion of the total load from MRTO transferred to the NMRTO. The ``Flowgate'' column \rteal{represents the transmission line selected as the flowgate and the contingency line in the parenthesis. It is} shown in the format as ``monitored line (outage line).'' The ``M2M Cost (\$)'' column reports the sum of costs in both MRTO and NMRTO. \ren{T}he ``Central Cost (\$)'' column reports the optimal cost solved by the centralized formulation \eqref{model:cen}. The ``Gap'' column reports \ren{the optimality gap, which is defined as follows:} 
\begin{equation}
    \ren{(\mbox{\it M2M Cost} - \mbox{\it Central Cost}) / \mbox{\it Central Cost} \times 100\%.}
\end{equation}
The ``Curtailment?'' column indicates whether the M2M Cost reported includes a curtailment penalty or not, \ren{which represents the penalty for the amount of load curtailed at all buses as reported in the UnitCommitment.jl package.} Finally, the ``Shadow Prices'' column indicates whether the shadow prices from the two RTOs are equal upon termination. If the two RTOs have the same shadow price on the flowgate on termination, it is marked as ``Converged''; otherwise, it is marked as ``Didn't converge''.

\begin{table*}[ht]
\centering
\caption{Iterative M2M Method Results}
\begin{tabular}{lccrrrcc}
\toprule
Network  & Interchange Ratio & Flowgate &M2M Cost (\$) & Central Cost (\$) & Gap & Curtailment?  & Shadow Prices   \\
\midrule
case1951rte & 0.03 & l1364 (l1102) &  4,118,587.63  &  4,118,780.35  & 0.00\% & No & Converged \\ 
case2383wp & 0.03 & l285 (l323) &  672,352.88  &  512,205.31  & 31.27\% & Yes & Converged \\ 
case2868rte & 0.03 & l2285 (l2288) &  4,364,815.79  &  3,746,785.86  & 16.49\% & Yes & Converged \\ 
case3120sp & 0.03 & l1267 (l2990) &  502,668.63  &  494,387.73  & 1.67\% & No & Converged \\ 
case3375wp & 0.03 & l614 (l615) &  959,276.46  &  959,276.45  & 0.00\% & No & Converged \\ 
case6468rte & 0.03 & l6911 (l7296) &  2,636,956.96  &  2,636,956.90  & 0.00\% & No & Converged \\ \midrule
case1951rte-of & 0.03 & l1777 (l2013) &  Inf  &  4,118,780.35  & Inf & - & - \\ 
case1951rte-ll & 0.04 & l1364 (l1102) &  4,271,516.10  &  4,239,576.94  & 0.75\% & Yes & Didn't converge \\ 
case2383wp-ll & 0.07 & l292 (l294) &  644,205.81  &  584,522.33  & 10.21\% & No & Didn't converge \\ 
case2868rte-ll & 0.03 & l2285 (l2288) &  4,373,150.63  &  3,779,859.13  & 15.70\% & Yes & Converged \\ 
case3120sp-ll & 0.03 & l1609 (l1731) &  509,531.49  &  495,746.67  & 2.78\% & Yes & Didn't converge \\ 
case3375wp-ll & 0.03 & l614 (l615) &  2,150,712.51  &  967,094.45  & 122.39\% & Yes & Didn't converge \\ 
case6468rte-ll & 0.03 & l6911 (l7296) &  2,674,160.46  &  2,642,145.72  & 1.21\% & No & Didn't converge \\
\bottomrule
\end{tabular}
\label{tab:iterative}
\end{table*}

From the table, we can \ren{observe} that three instances (case1951rte, case3375wp, and case6468rte) achieve $0\%$ gaps under the iterative method, and their shadow prices converge upon termination. However, the congestion in the remaining instances \ren{is} not completely mitigated, as they \ren{have} positive \ren{optimality} gaps. For example, the ``case3375wp-ll'' instance exhibits a $122.39\%$ gap due to load curtailment in the iterative method. Moreover, the iterative M2M method fails to provide a feasible solution for the ``opposite-flow'' instance ``case1951rte-of.'' We further study these under-mitigated instances to gain insights into the three challenges the current iterative method face\ren{s}.

\subsubsection{Power Swings} \label{sec:swing}

The first challenge is the ``power swing'' issue, where shadow prices from the two RTOs oscillate indefinitely without converging. Analyzing the iterations, we found that at one point, the NMRTO granted a relief request that exceeded the MRTO's actual need, causing the flowgate at the MRTO to no longer be binding. This led to a zero shadow price at the MRTO. When the NMRTO received this zero shadow price and adopted it as its shadow price limit, it essentially removed the flowgate constraint \rteal{in \eqref{model:iter} for NMRTO because $s_{\ell,2}$ amount is free,} and reverted its dispatch to the previous level, causing congestion again at the MRTO.
This cyclical process repeats, creating the ``power swing'' issue. 

\ren{One example for the power swing case, the ``case2383wp-ll'' network, is} illustrated in Fig.~\ref{fig:swing}. \ren{The} oscillation in shadow prices leads to fluctuations in flow on the flowgate for both RTOs. For example, in \rteal{I}teration $4$, the MRTO, after receiving the granted relief request, adjusts its flowgate limit to $126$ MW (indicated by the red dashed line). However, the new SCED with this updated limit dispatches from $-42$ MW to $-54$ MW on the MRTO side (indicated by the black line), which is not binding. Subsequently, the NMRTO, after receiving the zero shadow price \rteal{from MRTO, has a relaxed} flowgate constraint and \rteal{thus} increases its flow from $-109$ MW to $-188$ MW. This repetitive process results in intermittent congestion, which cannot be resolved by relying solely on the convergence of shadow prices.

\begin{figure}[ht]
    \centering
    \includegraphics[width=\columnwidth]{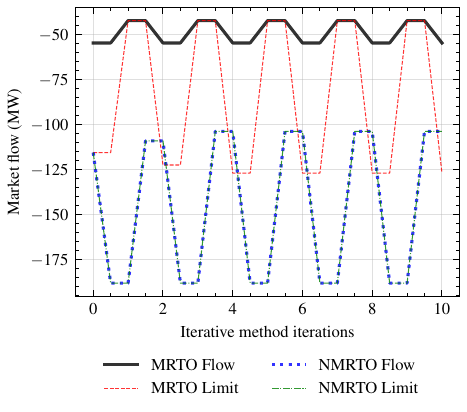}
    \caption{Power swing issue of the iterative method in the ``case2383wp-ll'' instance. The $x$-axis represents the number of iterations in the iterative M2M method, and the $y$-axis represents the flow amount in MW.}
    \label{fig:swing}
\end{figure}

Moreover, in real-world power systems, two systematic time delays can exacerbate this problem: First, the two RTOs do not run SCED synchronously, leading to SCED flows with different timestamps. Second, there is a delay between the SCED and the resource response needed to relieve M2M constraints. These delays cause inaccuracies in the relief request calculation, leading to more severe power swing issues than observed in our simulations.

\subsubsection{Non-Stopping Behavior Despite No Change in Solution}

The second issue concerns the stopping criterion, which relies on the shadow prices of the flowgate from the two RTOs becoming equal. 
However, shadow prices may remain mismatched even when the relief request is not changing anymore. For instance, for \rteal{the} ``case3375wp-ll'' \rteal{case}, the shadow prices for MRTO and NMRTO are \$$43.98$ and \$$0.00$, respectively\rteal{, obtained from the iterative method}, and they remain unchanged until the algorithm terminates due to reaching the maximum number of iterations. Since the iterative method terminates based on shadow price convergence, this issue can lead to endless iterations. This problem typically arises when the generation costs between the two RTOs differ, preventing them from reaching the same shadow price to mitigate congestion. Notably, the instances that achieve zero gaps in Table~\ref{tab:iterative} all converge to a shadow price of \$$0.00$, indicating less severe congestion. 

\subsubsection{Opposite Flows}

The third potential issue is the iterative method's difficulty in handling flowgates with opposite flows from two RTOs. When flows oppose each other, they partially cancel out, resulting in a lower net flow. Properly mitigating congestion in such cases requires one RTO to increase its flow to enhance this cancellation effect. However, the relief request formula \eqref{eqn:relief} described in \cite{chen2017address} or \eqref{eqn:relief2} fail\rteal{s} to convey this need accurately.

The ``opposite-flow'' instance ``case1951rte-of'' \rteal{has been generated and} demonstrates that the flowgate with opposite flows can lead to infeasibility under the iterative method. To illustrate this issue, \rteal{we} consider a flowgate with a capacity of $100$ MW. Suppose the NMRTO has a flow of $100$ MW, while the MRTO has a flow of $-200$ MW. The total flow is $-100$ MW, which is within the limit. If the MRTO anticipates the flow to become $-250$ MW, it initiates an M2M process using the iterative method. The relief request calculation \rteal{following equation \eqref{eqn:relief2}} would be \rteal{$\abs{\abs{-250 + 100} - 100}+0 = 50$ MW with $Adder=0$ for simplicity,} indicating that the MRTO wants the NMRTO to reduce its flow by $50$ MW. However, to effectively mitigate congestion, the MRTO actually needs the NMRTO to increase its flow. This discrepancy highlights the limitations of the current relief request calculations in scenarios with opposite flows.

\subsection{Evaluation of \ren{the} ADMM \ren{Approach}}

\ren{Now}, we \ren{report} our implementation of the ADMM algorithm for M2M coordination and demonstrate its effectiveness in mitigating congestion compared to the current iterative method using the same M2M instances.

\ren{T}he SCED model for each RTO \ren{is built and executed} separately using the augmented Lagrangian relaxation \eqref{model:alr}. After both models are solved in parallel, the flowgate flows and other parameters are updated. This ADMM process continues until the desired \ren{schedule} is reached.

\ren{The ADMM algorithm was implemented to test the same set of M2M instances shown in Table~\ref{tab:iterative}. T}he \ren{optimality} gaps between the ADMM costs and the centralized method \ren{are reported in Table~\ref{tab:admm}. The interchange ratio and flowgate for this experiment are the same as in Table~\ref{tab:iterative}, and thus omitted here for simplicity.} The ``ADMM Gap'' column is calculated as the difference between the ADMM Cost and the Central Cost divided by the Central Cost. With this column showing $0\%$ gaps, we can find that all instances \ren{reach the optimal solution by converging to} the lower bound established by the centralized method. Additionally, the ``Iterations'' column shows the total number of iterations used in the ADMM process, and the ``Time (Seconds)'' column shows the total time used for the ADMM process. 

\begin{table}[ht]
\centering
\caption{ADMM M2M Method Results}
\begin{tabular}{lccc}
\toprule
Network & ADMM Gap & Iterations & Time (Seconds) \\ \midrule
case1951rte & 0.00\% & 17 & 1.14 \\ 
case2383wp & 0.00\% & 45 & 1.70 \\ 
case2868rte & 0.00\% & 253 & 3.79 \\ 
case3120sp & 0.00\% & 60 & 1.62 \\ 
case3375wp & 0.00\% & 6 & 0.91 \\ 
case6468rte & 0.00\% & 109 & 6.71 \\\midrule
case1951rte-of & 0.00\% & 31 & 1.12 \\ 
case1951rte-ll & 0.00\% & 15 & 1.44 \\ 
case2383wp-ll & 0.00\% & 77 & 1.94 \\ 
case2868rte-ll & 0.00\% & 76 & 1.66 \\ 
case3120sp-ll & 0.00\% & 14 & 1.24 \\ 
case3375wp-ll & 0.00\% & 24 & 1.62 \\ 
case6468rte-ll & 0.00\% & 13 & 2.03 \\ 
\bottomrule
\end{tabular}
\label{tab:admm}
\end{table}

The \ren{table} show\ren{s} that our proposed ADMM M2M method addresses the aforementioned three issues with fast convergence. First, the ADMM approach tackles the power swing issue by converging to an optimal solution without the need to set the shadow price limit. This prevents the congestion status from relapsing thus avoiding the power swing issue. For example, as shown in Fig.~\ref{fig:admm}, the ``case2383wp-ll'' lower-limit instance, which experienced power swings under the iterative method, converged to an acceptable feasible solution in $77$ iterations in $1.94$ seconds.

\begin{figure}[ht]
    \centering
    \includegraphics[width=\columnwidth]{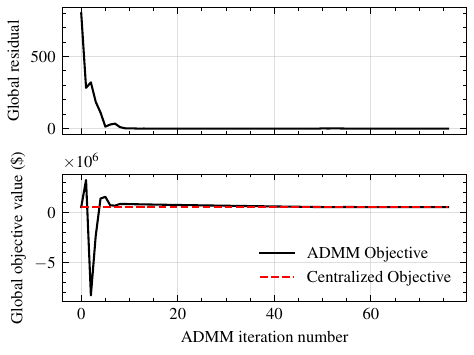}
    \caption{\ren{The} ADMM M2M results for \ren{the} ``case2383wp-ll'' instance. The $x$-axes in both subplots represent the iteration number\ren{s} in the ADMM algorithm, while the $y$-axis represents the global residual in the upper subplot and represents the total cost in the lower subplot. }
    \label{fig:admm}
\end{figure}

Second, the ADMM approach handles the non-stopping issue caused by the shadow price mismatch. The stopping criteria for the ADMM algorithm check the solution feasibility and objective value improvement, instead of the match of shadow prices. Following Proposition~\ref{prop:convergence}, the ADMM algorithm is guaranteed to stop upon an optimal solution and avoid endless iterations. As illustrated in Fig.~\ref{fig:admm-sp}, the shadow prices for both RTOs converge as the global objective value converges to the optimal objective value. The upper subplot illustrates that the ``case3375wp-ll'' instance, which previously faced a non-stopping issue under the iterative method, is resolved optimally using the ADMM approach. \ren{For this case, besides the shadow price converging for each RTO, the shadow prices from both RTOs are equal eventually.} The lower subplot depicts a more typical scenario where the shadow prices do not match even at the optimal M2M solution. \rteal{This observation indicates that} the stopping criteria of the iterative method \rteal{to make both sides of the shadow prices equal at termination cannot even be satisfied by an optimal solution. Meanwhile, it} highlights the advantage of our proposed ADMM approach.
\begin{figure}[ht]
    \centering
    \includegraphics[width=\columnwidth]{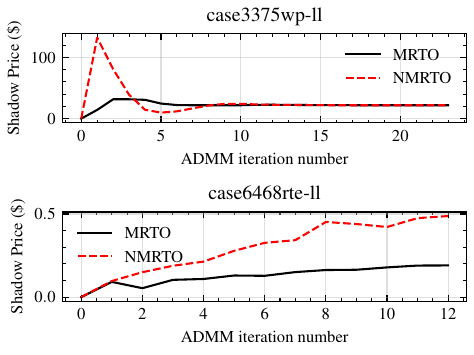}
    \caption{\ren{The} ADMM M2M results for \ren{the} ``case3375wp-ll'' and ``case6468rte-ll'' \ren{instances}. The $x$-axes in both subplots represent the iteration number\ren{s} in the ADMM algorithm, while the $y$-axis in both subplots represents the shadow price of the flowgate for two RTOs.}
    \label{fig:admm-sp}
\end{figure}

Third, the ADMM approach also solves the congested flowgate with opposite flows from \ren{both} RTOs. Since our proposed ADMM algorithm does not calculate an explicit relief request, it does not need to handle the opposite flows differently from \ren{cases where both RTOs have the same flow direction on the flowgate}. Thus, it guarantees to find the optimal congestion mitigation solution for both markets in the opposite flow situation.

\ren{In summary, from the above computational results, it can be observed that} our ADMM formulation for the M2M coordination guarantees a $0\%$ gap with the centralized method and effectively addresses the three issues highlighted in the current iterative method.

\section{Conclusion} \label{sec:conclusion}

In this paper, we addressed the M2M coordination problem by \ren{analyzing} the current iterative method used in practice and \ren{proposing} a decentralized ADMM algorithm to \ren{improve the overall performance}. 

\ren{Our ADMM approach can converge to an optimal solution, as compared to the current practice, which is a heuristic approach and cannot guarantee an optimal solution. Our computational case studies verify the effectiveness of our ADMM approach by showing it converges to a lower bound provided by solving a centralized model, which indicates an optimal solution is obtained. Accordingly, our approach resolves challenging issues, such as power swings, non-convergence, and opposite flow issues faced by the current practice.}

\ren{In future research, we will explore} insights from the ADMM results to improve \ren{current} M2M coordination, refin\ren{e} the iterative method by enhancing the relief request calculation, and extend the proposed approaches to multi-area coordination.

\section*{Acknowledgement}

This work was partially supported by the U.S. Department of Energy Advanced Grid Modeling Program under Grant DE-OE0000875.

\bibliographystyle{IEEEtran}

\bibliography{m2m}

\end{document}